\documentclass[a4paper,onecolumn,accepted=2022-08-18,12pt]{quantumarticle}

\pdfoutput=1

\usepackage{graphicx}
\usepackage{amsmath}
\usepackage{amssymb}
\usepackage{amsfonts}
\usepackage{amsthm}
\usepackage{mathtools}
\usepackage{hyperref}
\usepackage{url}
\usepackage{xcolor}
\usepackage[T1]{fontenc}
\usepackage[english]{babel}
\usepackage{braket}
\usepackage{enumitem}
\usepackage{tikz}
\usepackage{bbm}

\hypersetup{colorlinks=true,urlcolor=[rgb]{0,0,0.5},citecolor=[rgb]{0.5,0,0},linkcolor=[rgb]{0,0,0.4}}

\usepackage[bbgreekl]{mathbbol}
\DeclareSymbolFontAlphabet{\mathbb}{AMSb}
\DeclareSymbolFontAlphabet{\mathbbl}{bbold}

\newtheorem{theorem}{Theorem}
\newtheorem{corollary}{Corollary}
\newtheorem{lemma}{Lemma}
\newtheorem{observation}{Observation}

\newtheorem{definition}{Definition}

\def\autorefapp#1{\hyperref[#1]{Appendix~\ref{#1}}}

\def\ii{\mathrm{i}}

\newcommand{\jh}[1]{
	{#1}
}

\def\ep{\varepsilon}

\def\iden{\mathbbl{1}}

\def\dm{d}
\def\met{\mathfrak{g}}

\begin{document}

\title{Random quantum circuits are approximate unitary $t$-designs in depth $O\left(nt^{5+o(1)}\right)$}

\author{Jonas Haferkamp}
\affiliation{Dahlem Center for Complex Quantum Systems, Freie Universit{\"a}t Berlin, Germany}

\maketitle

\begin{abstract}
	The applications of random quantum circuits range from quantum computing and quantum many-body systems to the physics of black holes.
	Many of these applications are related to the generation of quantum pseudorandomness: Random quantum circuits are known to approximate unitary $t$-designs.
	Unitary $t$-designs are probability distributions that mimic Haar randomness up to $t$th moments.
	In a seminal paper, Brand\~{a}o, Harrow and Horodecki prove that random quantum circuits \jh{on qubits} in a brickwork architecture of depth $O(n t^{10.5})$ are approximate unitary \jh{$t$-designs}. 
	In this work, we revisit this argument, which lower bounds the spectral gap of moment operators for local random quantum circuits by $\Omega(n^{-1}t^{-9.5})$.
	We improve this lower bound to $\Omega(n^{-1}t^{-4-o(1)})$, where the $o(1)$ term goes to $0$ as $t\to\infty$.
	A direct consequence of this scaling is that random quantum circuits generate approximate unitary $t$-designs in depth $O(nt^{5+o(1)})$.
	Our techniques involve Gao's quantum union bound and the unreasonable effectiveness of the Clifford group.
	As an auxiliary result, \jh{we prove fast convergence} to the Haar measure for random Clifford unitaries interleaved with Haar random single qubit unitaries.
\end{abstract}

\noindent Random unitaries are a ubiquitous concept in quantum information theory and quantum many-body physics. 
This ranges from practical applications such as randomized benchmarking~\cite{EmeAliZyc05,MagGamEmer,KnillBenchmarking} to mixing in black holes~\cite{HaydenPreskill}.
However, the applicability of uniformly (Haar) random unitaries and states is limited by the fact that they require exponential resources as counting arguments show~\cite{knill1995approximation}.
It is therefore desirable to consider less complex probability distributions over the unitary group that are sufficiently random for practical purposes.
This leads naturally to the notion of unitary $t$-designs~\cite{dankert2005efficient,dankert_exact_2009,gross_evenly_2007,AmbEme07}.
These are probability distributions which mimic the Haar measure up to $t$th moments.

There is an ongoing effort to find efficient constructions of unitary $t$-designs \cite{cleve2015near,harrow2009random}.
In a seminal result, Brand\~{a}o, Harrow and Horodecki proved that local random quantum circuits are approximate unitary $t$-designs after the application of $O(n^2t^{10.5})$ random gates for all $t\leq O(2^{0.4n})$~\cite{brandao_local_2016}.
This result has several consequences and many results~\cite{harrow2018approximate,onorati2017mixing,nakata2017efficient,mezher2018efficient} on the statistical properties of random processes over the unitary group directly depend on the bound from Ref.~\cite{brandao_local_2016}.

Recently, unitary $t$-designs played a key role in lower bounding quantum circuit complexity~\cite{brandao_complexity_2019,brandao2016efficient}.
Here, Ref.~\cite{brandao_complexity_2019} shows that strong notions of quantum circuit complexity of a unitary $t$-design can be lower bounded by $\Omega(nt)$ with high probability.
Combined with the design depth from Ref.~\cite{brandao_local_2016}, this roughly implies that the complexity of most quantum circuit of depth $T$ can be lower bounded by $\Omega(T^{1/10.5})$.
This line of research is inspired by the Brown-Susskind conjecture~\cite{brown2018second} in the context of the AdS/CFT correspondence~\cite{susskind2018black}.
The conjecture states that the quantum circuit complexity grows at a linear rate for an exponentially long time.
Indeed, Ref.~\cite{brandao_complexity_2019} implies such a linear lower bound in the limit of large local dimension~\cite{hunter2019unitary} (it was later shown that a scaling of $\sim 6t^2$ for the local dimension is sufficient~\cite{haferkamp2021improved}).
More recently, a variant of the Brown-Susskind conjecture for random quantum circuits was proven for the special case of the exact circuit complexity~\cite{haferkamp2021linear}.
However, for error-robust (and therefore more operational) notions of quantum circuit complexity and for random quantum circuits over qubits, the best lower bound for exponentially long times is $\Omega(T^{1/10.5})$.

Here, we want to improve the scaling of the design depth in $t$.
Indeed, the scaling in $n$ is already optimal (up to logarithmic factors).
It was conjectured in Ref.~\cite{brandao_local_2016} that the true scaling of the design depth might be linear in $t$, which would directly imply an error-robust version of the Brown-Susskind conjecture~\cite{brandao_complexity_2019}. 

In this work, we dissect the bound obtained on the spectral gap of random quantum circuits in Ref.~\cite{brandao_local_2016} and improve the lower bound on this spectral gap from $\Omega(n^{-1}t^{-9.5})$ to $\Omega(n^{-1}t^{-4-o(1)})$.

The lion's share of the improvement is a new unconditional bound on the spectral gap.
We achieve this bound by considering an auxiliary random walk that interleaves global uniformly random Clifford unitaries with Haar random unitaries on a single qubit.
This walk was introduced in Ref.~\cite{haferkamp2020quantum,zhou2020single} to analyze the number of non-Clifford unitaries required to approximate unitary $t$-designs for small values of $t$ satisfying $O(t^2)\leq n$.
Indeed, this random walk has desirable properties for the generation of unitary designs.
It was proven that, while relative error approximations require $\theta(n)$ many non-Clifford gates~\cite{leone2021quantum,oliviero2021transitions}, an additive constant approximation of unitary designs requires only a system-size independent amount of non-Clifford resources~\cite{haferkamp2020quantum}.
We show that the path coupling technique on the unitary group~\cite{oliveira2009convergence} provides a \jh{fast convergence} to the Haar measure for this auxiliary random walk.
As the Clifford group is a finite group with upper bounded circuit complexity, the comparison technique from Ref.~\cite{diaconis1993comparison} provides fast mixing bounds, which allows us to approximate the uniform measure on the Clifford group with a local random walk.
We then translate these bounds to a spectral gap for local random quantum circuits by invoking Gao's quantum union bound~\cite{gao2015quantum,anshu2016simple}.

\section{Preliminaries}
A central object of this paper is the moment superoperator defined with respect to a probability distribution $\nu$ on the unitary group $U(\dm)$, defined as
\begin{equation}
\Phi_\nu^{(t)}(A):=\int U^{\otimes t}A(U^{\dagger})^{\otimes t}\, \mathrm{d}\nu(U)\,.
\end{equation}
We denote the Haar-measure on the unitary group by $\mu_H$.

We make use of the following isomorphism called vectorization:  $\mathrm{vec}:\mathbb{C}^{\jh{d\times d}}\to \mathbb{C}^{\jh{d^2}}$ with $\mathrm{vec}(|i\rangle\langle j|)=|i\rangle\otimes |j\rangle$.
This isomorphism extends to superoperators, mapping them to matrices: $\mathrm{vec}(T)\mathrm{vec}(M):=\mathrm{vec}(T(M))$ for all $M\in \mathbb{C}^{\jh{d\times d}}$ for a superoperator $T$.

We introduce the following notation for the second highest eigenvalue of the moment operator:
\begin{equation}
g(\nu,t):=\big\|M(\nu,t)-M(\mu_H,t)\big\|_{\infty}\,,
\end{equation}
where the $t$-th moment operator of a probability distribution is defined as
\begin{equation}
M(\nu,t):=\mathrm{vec}\big(\Phi_\nu^{(t)}\big)=\int U^{\otimes t}\otimes \overline{U}{}^{\otimes t}\,\mathrm{d}\nu(U)
\end{equation} 
and $\|\bullet\|_{\infty}$ denotes the Schatten $\infty$-norm. 
The \textit{spectral gap} of the moment operator is $1-g(\nu,t)$.
Notice that we defined the spectral gap to be the difference between the first $t!$ eigenvalues and the $(t!+1)$st eigenvalue of the moment operator. 
That means that for non-universal probability distributions the gap can be $0$.
A probability distribution $\nu$ is called a \textit{$(\lambda,t)$-tensor product expander} if the spectral gap of the moment operator satisfies $1-g(\nu,t)\jh{\leq} \lambda$.

Moreover, we denote the Schatten $1$-norm by $||\bullet||_p$ and the Schatten $2$-norm/Frobenius norm by $||\bullet||_F$.
The stabilized induced (Schatten) $1$-norm -- called \textit{diamond norm} -- of superoperators will be denoted by $||\bullet||_{\diamond}$~\cite{aharonov1998quantum,watrous2018theory}.
Roughly, the distance between channels in diamond norm quantifies the ``one-shot'' distinguishability using entangled input states.

The second highest eigenvalue $g(\nu,t)$ can be amplified.
Specifically, the $k$-fold convolution of $\nu$ has the property that
\begin{equation}\label{eq:gconvolution}
g(\nu^{*k},t)\leq g(\nu,t)^k\,.
\end{equation}
Upper bounds on the second highest eigenvalue can be used to imply an approximate version of unitary designs~\cite{brandao_local_2016}.
We define approximate designs in two (inequivalent) ways, with a relative error and with an exponentially small additive error. 
The relation to the spectral gap turns out to be the same.
\begin{definition}[Approximate unitary designs] \label{def:approxdesign}
	\textcolor{white}{a}
	\begin{enumerate}
		\item A probability distribution $\nu$ on $U(\dm)$ is an $\varepsilon$-approximate unitary $t$-design if the moment superoperator obeys 
		\begin{equation} 
		\big\|\Phi_\nu^{(t)}-\Phi_{\mu_H}^{(t)}\big\|_\diamond \leq \frac{\varepsilon}{\dm^t}\,.
		\label{eq:approxdesign1}
		\end{equation}
		\item A probability distribution $\nu$ on $U(\dm)$ is a relative $\varepsilon$-approximate $t$-design if 
		\begin{equation}
		(1-\varepsilon)\Phi_\nu^{(t)} \preccurlyeq \Phi_{\mu_H}^{(t)} \preccurlyeq	(1+\varepsilon)\Phi_\nu^{(t)}\,,
		\label{eq:approxdesign2}
		\end{equation}
		where here $A\preccurlyeq B$ if and only if $B-A$ is a completely positive map.
	\end{enumerate}
\end{definition}
By applying \cite[Lemma~4]{brandao_local_2016}, as well as the fact that $\|\Phi^{(t)}_\nu - \Phi^{(t)}_{\mu_H}\|_\diamond \leq \dm^t g(\nu,t)$~\cite{LowThesis}, we can state:
\begin{lemma}\label{lemma:gtodesign}
	Let $\nu$ be a probability distribution on $U(\dm)$ such that $g(\nu,t)\leq \varepsilon/\dm^{2t}$. 
	Then $\nu$ is an $\ep$-approximate unitary $t$-design and obeys both Eq.~\eqref{eq:approxdesign1} and Eq.~\eqref{eq:approxdesign2}.
\end{lemma}

In this paper we consider the following architectures of random quantum circuits comprised of 2-local unitary gates on a system of $n$ qubits.
\begin{definition}[Random quantum circuits] \label{def:RQCs}
	\textcolor{white}{a}
	\begin{enumerate}
		\item \emph{Local ($1D$) random quantum circuits:} Let $\nu_n$ denote the probability distribution on $U((\mathbb{C}^{2})^{\otimes n})$ defined by first choosing a random pair of adjacent \jh{qubits} and then applying a Haar random unitary $U_{i,i+1}$ from $U(4)$.
	     We assume periodic boundary conditions.
		\item \emph{Brickwork random quantum circuits:} Apply first a unitary $U_{1,2}\otimes U_{3,4}\otimes...$ and then a unitary $U_{2,3}\otimes U_{4,5}\otimes...$, where all $U_{i,i+1}$ are drawn Haar-randomly. For simplicity we assume in this case an even number of \jh{qubits}.  
		We denote this distribution by $\nu^{\rm bw}_n$.
	\end{enumerate}
\end{definition}

A standard tool for converting results on the gap on local random quantum circuits into brickwork circuits is the detectability lemma~\cite{aharonov2009detectability,anshu2016simple}. 
We will also make use of a strong converse called the \textit{quantum union bound}~\cite{gao2015quantum,anshu2016simple}.
We present both together:
\begin{lemma}[Detectability lemma and union bound]\label{lemma:detectabilityandconverse}
	Let $H=\sum_i Q_i$, where $Q_i$ are orthogonal projectors in an arbitrary order.
	Assume that each $Q_i$ commutes with all but $g$ of the projectors. 
	\jh{Moreover, assume that $H$ is frustration-free, i.e. a state $|\psi\rangle$ with $\langle \psi|H|\psi\rangle=0$ exists.}
	Then, for a state $|\psi^{\perp}\rangle$ orthogonal to the ground space of $H$,
	\begin{equation}\label{eq:detectability}
	\sqrt{1-4\Delta(H)}\leq \left|\left|\prod_{i} (\iden-Q_i)|\psi^{\perp}\rangle\right|\right|_2\leq \sqrt{\frac{1}{\Delta(H)/g^2+1}}.
	\end{equation}
\end{lemma}

Finally, we introduce the Clifford group, which plays a crucial role in the proof of Theorem~\ref{thm:unconditionalgap}:
The $n$-qubit Clifford group $\mathrm{Cl}(n)$ is the unitary normalizer of the Pauli group $\mathcal{P}_n$:
\begin{equation}
\label{eq:clifford-group}
\mathrm{Cl}(n) := \left\{ U \in U(2^n, \mathbb{Q}[i]) \; \big| \; U\mathcal{P}_n U^\dagger \subset \mathcal{P}_n \right\}.
\end{equation}
Our interest stems from the fact that this group forms a unitary $2$-design~\cite{Webb15,Zhu15}, while every element has a polynomially upper bounded circuit complexity~\cite{aaronson2004improved,bravyi2021hadamard}.

\section{Revisiting the spectral gap of random quantum circuits}
Here we review the key steps in the proof technique in Ref.~\cite{brandao_local_2016}. 
We introduce the notation $P_{H}:=M(\mu_H,t)$ and $P_{H,m}$ to specify that we are talking about the Haar moment operator on $m$ qubits.
Ref.~\cite{brandao_local_2016} introduced the frustration-free Hamiltonian
\begin{equation}
H_{n,t}:=\sum_{i=1}^{n}(\iden-\iden_{i-1}\otimes P_{H,2}\otimes \iden_{n-i-1}).
\end{equation}
With this definition it holds that~\cite[Lem.~16]{brandao_local_2016}:
\begin{equation}\label{eq:gapreduction}
g(\nu_n,t)=1-\frac{\Delta(H_{n,t})}{n},
\end{equation}
where $\Delta(H_{n,t})$ denote the spectral gap of $H_{n,t}$.

Using Nachtergaele's martingale method for spin systems, Ref.~\cite{brandao_local_2016} proves that for all $n$ and $t$ such that $n\geq \lceil 2.5 \log_2(4t)\rceil$ the following bound holds:
\begin{equation}\label{eq:BHHgapreduction}
\Delta(H_{n,t})\geq \frac{\Delta(H_{\lceil 2.5 \log_2(4t)\rceil,t})}{4 \lceil 2.5 \log_2(4t)\rceil}.
\end{equation}
Remarkably, this reduction not only establishes a constant spectral gap for every fixed value of $t$, but also turns every bound that only depends exponentially on the system size into one that depends polynomially on $t$.

The reduction~\eqref{eq:BHHgapreduction} to smaller system sizes is then combined with an spectral gap that holds for all values of $t$.
We call such a gap \textit{unconditional}:
\begin{equation}\label{eq:BHHunconditionalgap}
\Delta(H_{n,t})\geq \frac{1}{n(5e)^n}.
\end{equation}
This bound is obtained from a convergence result: By applying the path coupling technique on the unitary group~\cite{oliveira2009convergence}, the authors of Ref.~\cite{brandao_local_2016} show that random quantum circuits converge to the Haar measure in the Wasserstein distance.
This convergence result is strong enough to imply the bound~\eqref{eq:BHHunconditionalgap} on the spectral gap.

How does the exponent of $t$ for the design depth in Ref.~\cite{brandao_local_2016} decompose exactly? 
Ignoring logarithmic factors, the exponent is approximately
\begin{equation}\label{eq:exponent_dissected}
10.41~~\approx~~~\underbrace{1}_{1.}\quad+\quad \left(\underbrace{2}_{2.}\quad+\quad \underbrace{0.5}_{3.}\right)\quad \times \quad \left(\underbrace{\log_2{5}}_{4.}\quad+\quad\underbrace{\log_2(e)}_{5.}\right).
\end{equation}
\begin{enumerate}
	\item This contribution comes from the conversion of a spectral gap to the stronger notions of approximate designs in Definition~\ref{def:approxdesign}.
	This contribution is therefore a necessary consequence of using spectral gaps to bound the design depth.
	\item Ref.~\cite{brandao_local_2016} establishes approximate orthogonality of permutations of $t$ tensor factors of $\mathbb{C}^D$ in the regime $D\geq t^2$.
	In the application of Nachtergaele's martingale method~\cite{nachtergaele1996spectral}, this fact is repeatedly used and the second contribution can be directly linked to the square in the condition $t^2\leq D$.
	More precisely, the following inequality holds
	\begin{equation}\label{eq:BHHinequality}
	\left|\left|M(\mu_H,t)-2^{-tn}\sum_{\pi\in S_t}\mathrm{vec}(r(\pi))\mathrm{vec}(r(\pi))^{\dagger}\right|\right|_{\infty}\leq \frac{t^2}{D},
	\end{equation}
	where $r$ denotes the representation of the symmetric group that permutes the $t$ tensor factors.
	The second operator in Eq.~\eqref{eq:BHHinequality} is called the frame operator of the permutations.
	It has the same eigenvalues as the Gram matrix of the permutations and is also the moment operator of random matrices with i.i.d. Gaussian entries by Wick's theorem.
    The approximate orthogonality of the permutations is therefore equivalent to the moments of the Haar measure being approximately equal to those of Gaussian matrices with i.i.d. entries. 
    Each monomial of degree $t$ only contains information about some $t\times t$ submatrix.
    Indeed, small submatrices $t\sim o(\sqrt{D})$ of Haar random unitaries are conjectured to be close to an i.i.d. Gaussian ensemble in total variation distance~\cite{aaronson2011computational,jiang2006many}.
    However, as the entries of Haar random matrices are correlated, large submatrices cannot be close to an i.i.d Gaussian ensemble.
    Consequently, we cannot expect to make the RHS of Eq.~\eqref{eq:BHHinequality} scale better than $t/D$.
	\item This contribution is a consequence of Nachtergaele's bound.
	More precisely, the final bound of Ref.~\cite[Lem.~18]{brandao_local_2016} is 
	\begin{equation}
	\Delta(H_{n,t})\geq \frac{\Delta(H_{l,t})}{4 l},
	\end{equation}
	for any $l$ that satisfies the inequality
	\begin{equation}\label{eq:BHHreductioninequality}
	\frac{6 t^2}{2^l}\leq \frac12 l^{-\frac12}.
	\end{equation}
	The choice $l=2.5 \lceil \log_2(4t)\rceil$ can be seen to always satisfy~\eqref{eq:BHHreductioninequality}.
    This bound gives a comparably simple expressions at the expense of a slightly worse than optimal exponent in $t$.
	However, as the objective of this work is to pedantically optimize the exponent of $t$, we use modern bounds on Lambert's W function to obtain such a tight bound in $l$.
	We present the resulting bound below in Observation~\ref{observation:improvedreduction}.
	\item The last two contributions are a consequence of the unconditional spectral gap in Eq.~\eqref{eq:BHHunconditionalgap}.
	In general, contribution 4. comes from an exponential decay $(q^2+1)^{-n}$ (for qubits $q=2$) of the unconditional bound on the gap.
	This part of the bound~\eqref{eq:BHHunconditionalgap} can be proven for random ``staircase circuits'' of the form $U_{1,2} U_{2,3}\ldots U_{n-1,n}$, with Haar random $2$-local unitaries $U_{i,i+1}$.
	The same bound holds for permutations of the unitaries $U_{i,i+1}$.
	In this work, we observe that the application of the path coupling method in Ref.~\cite{brandao_local_2016} uses the full Haar randomness of the first gate $U_{1,2}$ only.
	For all other gates, we merely need the second moments of the Haar measure and we would hence obtain the same bound if the gates $U_{2,3},\ldots,U_{n-1}$ are drawn uniformly from the Clifford group. 
	The key idea in our improved bound is to go one step further and directly apply the full $n$-qubit Clifford group and only worry about locality once this convergence result is translated \jh{to} spectral gaps.
	\item Permutations of these staircase random walks appear with non-zero probability after $n$ steps of local random quantum circuits. 
	More precisely, these are realized with a probability of $n!/n^n$.
	This explains the appearance of Euler's constant $e$ as Stirling's approximation of the factorial is $n!\sim (n/e)^n$.
	By applying Gao's union bound directly to the moment operators of Clifford brickwork circuits approximating the uniform distribution on the Cliffords, we can circumvent this counting step entirely. 
\end{enumerate}

We observe that the bound~\eqref{eq:BHHgapreduction} can be slightly improved.
A detailed argument for the following can be found in Appendix~\ref{app:improvedreduction}.
\begin{observation}[Improved gap reduction]\label{observation:improvedreduction}
	For all $n\geq \lceil 2\log_2(4t)+1.5\sqrt{\log_2(4t)}\rceil$ the following bound holds:
	\begin{equation}\label{eq:improvedreduction}
	\Delta(H_{n,t})\geq \frac{\Delta\left(H_{\lceil 2\log_2(4t)+1.5\sqrt{\log_2(4t)}\rceil,t}\right)}{4\left\lceil 2\log_2(4t)+1.5\sqrt{\log_2(4t)}\right\rceil}.
	\end{equation}
\end{observation}

\section{Results}
\subsection{Tensor product expanders and unitary designs}
Our first result is a bound on the spectral gap of the moment operator of local random quantum circuits that is independent of the moment $t$:

\begin{theorem}[Unconditional gap]\label{thm:unconditionalgap}
	We have the following bound for all $t\geq 1$:
	\begin{equation}
	g(\nu_n,t)\leq 1-120000^{-1}n^{-5} 2^{-2n}
	\end{equation}
\end{theorem}
This can be combined with the slightly improved reduction in Observation~\ref{observation:improvedreduction} to yield the following bounds on the tensor product expander properties of local random quantum circuits:
 \begin{theorem}[Tensor product expanders]\label{thm:finalgapbounds}
  We have the following bounds for $n\geq \lceil 2\log_2(4t)+1.5\sqrt{\log_2(4t)}\rceil$:
  \begin{itemize}
  	\item $g(\nu_n,t)\leq 1- \left(Cn\ln^5(t)t^{4+3\frac{1}{\sqrt{\log_2(t)}}}\right)^{-1}$
  	\item $g(\nu^{\mathrm{bw}}_n,t)\leq  1-\left(3C\ln^5(t)t^{4+3\frac{1}{\sqrt{\log_2(t)}}}\right)^{-1}$,
  \end{itemize}
where the constant can be taken to be $C= 10^{13}$.
 \end{theorem}
	\begin{proof}
	The bound on $g(\nu_{n},t)$ follows from plugging the unconditional bound
	\begin{equation}
	\Delta(H_{m,t})\geq 120000^{-1}m^{-4}2^{-2m}
	\end{equation}
	from Theorem~\ref{thm:unconditionalgap} into the reduction~\eqref{eq:improvedreduction} in Observation~\ref{observation:improvedreduction}.
	The bound on $g(\nu_n^{\mathrm{bw}})$ can be obtained from this by applying the detectability lemma (Lemma~\ref{lemma:detectabilityandconverse}):
	\begin{equation}
	g(\nu^{\mathrm{bw}}_n)\leq \sqrt{\frac{1}{\Delta(H_{n,t})/4+1}}\leq 1-\left(3C\ln^5(t))t^{4+3\frac{1}{\sqrt{\log_2(t)}}}\right)^{-1},
	\end{equation}
	which is the second part of Theorem~\ref{thm:finalgapbounds}.
\end{proof}
By Lemma~\ref{lemma:gtodesign}, this immediately implies the following corollary, which is the main application of Theorem~\ref{thm:finalgapbounds}:
\begin{corollary}[Unitary designs]\label{cor:unitarydesigns}
	For $n\geq \lceil 2\log_2(4t)+1.5\sqrt{\log_2(4t)}\rceil$, the following bounds hold
	\begin{itemize}
		\item Local random quantum circuits are $\varepsilon$-approximate unitary $t$-designs after 
		\begin{equation}\label{eq:designscaling}
		k\geq Cn\ln^5(t)t^{4+3\frac{1}{\sqrt{\log_2(t)}}}(2nt+\log_2(1/\varepsilon))
		\end{equation}
		steps.
		\item Brickwork quantum circuits are $\varepsilon$-approximate unitary $t$-designs after 
		\begin{equation}
		k\geq C\ln^5(t)t^{4+3\frac{1}{\sqrt{\log_2(t)}}}(2nt+\log_2(1/\varepsilon))
		\end{equation}
		steps.
	\end{itemize} 
\end{corollary}
In comparison to the exponent in Eq.~\eqref{eq:exponent_dissected}, the new exponent decomposes as follows.
\begin{equation}
5+3\frac{1}{\sqrt{\log_2(t)}}~~=~~~\underbrace{1}_{1.}\quad+\quad \left(\underbrace{2}_{2.}\quad+\quad \underbrace{1.5\frac{1}{\sqrt{\log_2(t)}}}_{3.}\right)\quad \times \quad \underbrace{2}_{4.},
\end{equation}
where we again ignored the log-factors in~\eqref{eq:designscaling}.
Here, contribution $1.$ and $2.$ appear for the same reason as $1.$ and $2.$ in Eq.~\eqref{eq:exponent_dissected}.
Contribution 3. is a direct consequence of Observation~\ref{observation:improvedreduction} and 4. follows from the improved bound on the unconditional gap in Theorem~\ref{thm:unconditionalgap}.

By applying an argument from Ref.~\cite{brandao_local_2016}, the same scaling can be implied for random quantum circuits drawn from a discrete invertible gate set with algebraic entries, using a powerful result by Bourgain and Gamburd~\cite{bourgain2012spectral}.
This bound comes with an implicit constant depending on the gate set.
Moreover, at the expense of additional polylogarithmic factors in $t$, the assumption of algebraic entries can be dropped as proven in Ref.~\cite{oszmaniec2021epsilon} using on a theorem from Ref.~\cite{varju2012random}.
Similarly, it was shown recently that the assumption of invertibility can be dropped at the expense of an additional factor of $n$~\cite{oszmaniec2021epsilon}.

\jh{In this paper we focus on random quantum circuits over qubits.
	This is because there are very concrete bounds available for the maximal circuit complexity of the multiqubit Clifford group~\cite{aaronson2004improved,bravyi2021hadamard}.
	However, the same proof strategy works for random quantum circuits defined over every local dimension that is prime because in these dimensions the Clifford group forms a unitary $2$-design.
	Therefore, all results in this paper can be proven for every prime power dimension with constants depending on the maximal circuit complexity of the multiqudit Clifford group. }

\subsection{Quantum circuit complexity}
We apply our result to the growth of circuit complexity~\cite{brandao_complexity_2019,brandao2016efficient}.
Let $\mathcal{G}$ denote a universal gate set of $2$-local gates.
Moreover, denote by $M_{\mathcal{G},R}$ the set of all unitaries that can be realized by concatenation of at most $R$ gates in $\mathcal{G}$.
\begin{definition}[Quantum circuit complexity]
	For $\delta\in(0,1]$ a
	\begin{itemize}
		\item 
 state $|\psi\rangle $ has $\delta$-state complexity
 \begin{equation}
 \mathcal{C}_{\delta}(|\psi\rangle):=\min\{R, \exists V\in M_{\mathcal{G},R}~\text{with}~\frac12 ||V|0^n\rangle\langle 0^n|V^{\dagger}-|\psi\rangle\langle\psi|||_1\leq \delta\}.
 \end{equation}
 \item unitary $U$ has $\delta$-unitary complexity
 \begin{equation}
 C_{\delta}(U):=\min\{R, \exists V\in M_{\mathcal{G},R}~\text{with}~|||\mathsf{V}-\mathsf{U}||_{\diamond}\leq \delta\},
 \end{equation}
 where $\mathsf{U}$ is the channel defined by $\mathsf{U}(\rho)=U\rho U^{\dagger}$ and $\mathsf{V}$ likewise.
 \end{itemize}
\end{definition}
We can use the following bound:
\begin{theorem}[Informal, Ref.~\cite{brandao_local_2016}]
	Let $\nu$ be a relative approximate unitary $t$-design for some $t\leq O(2^{n/2})$.
	Then, a unitary $U$ drawn from $\nu$ satisfies $\mathcal{C}_{\delta}(U)\geq \Omega(nt)$ and $\mathcal{C}_{\delta}(U|\psi\rangle)\geq \Omega(nt)$ with high probability.
\end{theorem}
Combining this theorem with Corollary~\ref{cor:unitarydesigns} yields the following lower bounds on the circuit complexity:
\begin{corollary}[Growth of quantum circuit complexity]\label{cor:growthcomplexity}
	Let $U$ be drawn from a random quantum circuit in brickwork architecture of depth $T$.
	Moreover, let $\delta\in (0,1)$ be constant in the system size.
	 Then, with \jh{probability $1-e^{-\Omega(n)}$}, we have for the
	\begin{itemize}
		\item $\delta$-state complexity $$\mathcal{C}_{\delta}(U0^n\rangle)\geq \Omega\left(T^{1/(5+o(1))}\right)$$
		\item $\delta$-unitary complexity $$\mathcal{C}_{\delta}(U)\geq \Omega\left(T^{1/(5+o(1))}\right)$$
	\end{itemize}
until $T\geq \Omega(2^{n/(2+o(1))})$.
\end{corollary}

For the error-robust notions of quantum circuit complexity defined above Corollary~\ref{cor:growthcomplexity} provides the strongest \jh{known} bounds for random quantum circuits over qubits of superpolynomial depth.

Ref.~\cite{brandao_complexity_2019} in fact provides \jh{an even} stronger result.
 Even an operational notion of complexity that quantifies the resources necessary to distinguish a unitary from the maximally depolarizing channel (the ``most useless'' channel) is lower bounded by $\Omega(T^{1/(5+o(1))})$.

\section{Proof of Theorem~\ref{thm:finalgapbounds} and Theorem~\ref{thm:unconditionalgap}}
Denote by $\mu_{H,1}$ the Haar measure on the subgroup $U(2)\otimes \mathbbm{1}_{n-1}$ and $\mu_{\mathrm{Cl}}$ the uniform measure on the Clifford group.
We apply the path coupling technique~\cite{bubley1997path} for the unitary group developed in Ref.~\cite{oliveira2009convergence} to the following auxiliary random walk that is defined by
\begin{equation}\label{eq:auxiliarywalk}
\sigma^{*k}:=(\mu_{\mathrm{Cl}}*\mu_{H,1}*\mu_{\mathrm{Cl}})^{*k}.
\end{equation}
This walk was defined in~\cite{haferkamp2020quantum} and shown to generate unitary $t$-designs in depth $k=O(t^{4})$ for ``small'' $t$.
More precisely, the result holds for the regime $t^2\leq O(n)$ and for a weaker definition of approximate unitary design.
\begin{figure}[h]
	\center
	\includegraphics[scale=1.17]{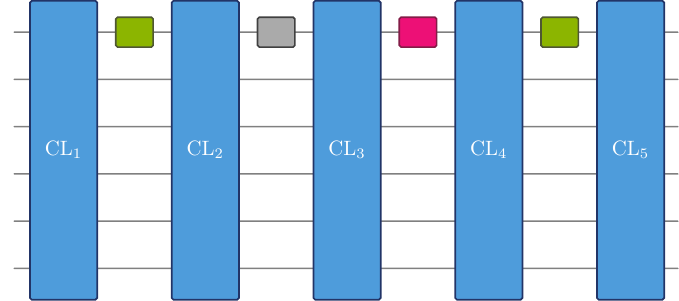}
	\caption{Four steps of the random walk $\sigma$ on $n=6$ qubits.}
\end{figure}

The first lemma is an unconditional spectral gap on the random walk $\sigma$:
\begin{lemma}\label{lemma:auxiliarywalk}
	We have for all $t\geq 1$ and $n\geq 1$ that
	\begin{equation}
	g(\sigma,t)\leq 1-\frac32 \frac{1}{2^{2n}-1}
	\end{equation}
\end{lemma}

The path coupling technique provides a bound on the convergence to the Haar measure in the Wasserstein distance, which is based on the notion of a random coupling: 
A random variable $(X,Y)$ is called a coupling for the distributions $\nu_1$ and $\nu_2$ if the marginal distributions for random variables $X$ and $Y$ are $\nu_1$ and $\nu_2$, respectively. 
This allows us to define the Wasserstein distance:
\begin{equation}
W_{\mathfrak{g},p}(\nu_1,\nu_2):=\inf \{(\mathbb{E}~\mathfrak{g}(X,Y)^p)^{1/p},\quad (X,Y)~\text{is a random coupling}\}.
\end{equation}

In Ref.~\cite{brandao_local_2016} a convergence result is proven for for the distance $W_{\mathrm{Rie},2}$.
Here, $\met_{\mathrm{Rie}}$ is the metric induced by the unique Riemannian metric that is invariant under left and right translations.
For our purposes, however, it is sufficient to prove a convergence in the weaker distance $W_{\mathrm{Fro},2}$.
We remark that the stronger convergence result can be obtained from our analyzis as in Ref.~\cite{brandao_local_2016}.

It was proven in Ref.~\cite{oliveira2009convergence} that infinitesimal contractions can be boosted to a global contraction of a random walk.
More precisely, Theorem~2 in Ref.~\cite{oliveira2009convergence} contains the following lemma as a special case:
\begin{lemma}\label{lemma:BHHsummaryofpathcoupling}
	Suppose that 
	\begin{equation}
	\limsup_{X\to Y}\left\{\frac{W_{\mathrm{Fro},2}(\nu*\delta_X,\nu*\delta_Y)}{||X-Y||_{F}}:||X-Y||_F\leq \varepsilon\right\}\leq \eta,
	\end{equation}
	for some probability distribution $\nu$ over $SU(d)$.
	Then, for all probability measures $\mu_1$ and $\mu_2$,
	\begin{equation}
W_{\mathrm{Fro},2}(\nu*\mu_1,\nu*\mu_2)\leq \eta W_{\mathrm{Fro},2}(\mu_1,\mu_2).
	\end{equation}
\end{lemma}

Using Lemma~\ref{lemma:BHHsummaryofpathcoupling}, we prove the following convergence result:
\begin{lemma}\label{lemma:wassersteinbound}
	For every $k\geq 1$, we have the bound:
	\begin{equation}
		W_{\met_{\mathrm{Fro}},2}\big(\sigma^{*k},\mu_{H}\big)\leq \left(1-\frac{3}{2^{2n}-1}\right)^{k/2}\pi 2^{n/2}.
	\end{equation}
\end{lemma}

\begin{proof}
	
A single step of the random walk $\sigma$ applied to a fixed unitary $X\in U(2^{n})$ by left multiplication yields:
\begin{equation}
X\to \left\{C'(U\otimes \mathbbm{1}_{n-1})C X\right\},
\end{equation}
for $U\in U(2)$ and $C,C'\in \mathrm{Cl}(n)$.
The unitary $Y$ undergoes the same transformation.
We can now introduce a family of random couplings for $(\sigma*\delta_X,\sigma*\delta_Y)$:
\begin{equation}
X'= C'(UV\otimes \mathbbm{1}_{n-1})C X\qquad\qquad Y'=C'(U\otimes \mathbbm{1}_{n-1})C X
\end{equation}
defined for a unitary $V\in U(2)$ that is independent of $U$ but can depend on $X,Y$ and $C$.

We need to bound:
\begin{align}
\mathbb{E}||X'-Y'||_{\mathrm{F}}^2=\mathbb{E}||C'(UV\otimes \mathbbm{1}_{n-1})C X-C'(U\otimes \mathbbm{1}_{n-1})C Y||_{\mathrm{F}}^2.
\end{align}
We can choose $V$ to be minimal for all $C,X$ and $Y$.
Then, we find
\begin{align}
\begin{split}
\mathbb{E}\left(\min_{V}||C'(UV\otimes \mathbbm{1}_{n-1})C X-C'(U\otimes \mathbbm{1}_{n-1})C Y||_{\mathrm{F}}^2\right)&=2\left(\mathrm{Tr}[\mathbbm{1}_n]-\max_{V}\mathrm{Tr}[V\otimes \mathbbm{1}_{n-1} CXY^{\dagger}C^{\dagger}]\right)\\
&=2\left(\mathrm{Tr}[\mathbbm{1}_n]-\mathbb{E}||\mathrm{Tr}_{[2,n]}(CXY^{\dagger}C^{\dagger})||_1\right),
\end{split}
\end{align}
where we used the variational characterization of the Schatten $1$-norm.
We can choose $X$ and $Y$ to be infinitesimally close together and write
\begin{equation}
R:=XY^{\dagger}=\exp(\ii \varepsilon H)=\mathbbm{1}_n+\ii \varepsilon H-\frac{\varepsilon^2}{2}H^2+O(\varepsilon^3),
\end{equation}
with $||H||_{\mathrm{F}}\leq 1$.
By Taylor expanding the $1$-norm, this implies (as in Ref.~\cite{brandao_local_2016}):
\begin{multline}
\mathbb{E}\left(\min_{V}||C'(UV\otimes \mathbbm{1}_{n-1})C X-C'(U\otimes \mathbbm{1}_{n-1})C Y||_{\mathrm{F}}^2\right)\\=\varepsilon^2\left(\mathrm{Tr}(H^2)-\frac{1}{2^{n-1}}\mathbb{E}\left[\mathrm{Tr}\left(\mathrm{Tr}_{[2,n]}\left[CHC^{\dagger}\right]\right)^2\right]\right)+O(\varepsilon^3).
\end{multline}
\jh{Denote by $\mathbb{F}$ the swap operator defined by $\mathbb{F}:\mathbb{C}^d\otimes\mathbb{C}^d\to \mathbb{C}^d\otimes\mathbb{C}^d$, $\mathbb{F} |\psi\rangle\otimes |\phi\rangle=|\phi\rangle\otimes |\psi\rangle$.}
The ``replica trick'' 
\begin{equation}
\mathrm{Tr}[A^2]=\mathrm{Tr}[A^{\otimes 2}\mathbb{F}]
\end{equation}
for every matrix $A$
 can be applied to show
\begin{align}
\begin{split}
\mathbb{E}\left[\mathrm{Tr}\left(\mathrm{Tr}_{[2,n]}\left[CHC^{\dagger}\right]\right)^2\right]&= \mathrm{Tr}\left[(CHC^{\dagger})^{\otimes 2}~ \mathbb{F}_1\otimes \mathbbm{1}_{[2,n]}\right]\\
&=\mathrm{Tr}\left[H^{\otimes 2}~ (C^{\dagger})^{\otimes 2}(\mathbb{F}_1\otimes \mathbbm{1}_{[2,n]})C^{\otimes 2}\right].
\end{split}
\end{align}
We now use that the Clifford group is an exact unitary $2$-design: the expectation value in the above equation is the same as that for Haar-random unitaries $C$.
We employ standard formulas~\cite[Lem.~IV.3]{abeyesinghe2009mother} for the twirl over the Haar measure:
    \begin{align}
\begin{split}
\mathbb{E}_{C\sim\mu_{\mathrm{Cl}}}~&\mathrm{Tr}[H^{\otimes 2}~ (C^{\dagger})^{\otimes 2}(\mathbb{F}_1\otimes \mathbbm{1}_{[2,n]})C^{\otimes 2}]\\
&=\mathrm{Tr}\left[H^{\otimes 2}\left(\frac{2+2^{n-1}}{2^n+1}\frac12(\iden+\mathbb{F})-\frac{2-2^{n-1}}{2^n-1}\frac12(\iden-\mathbb{F})\right)\right]\\
&=\frac12\left(\frac{2+2^{n-1}}{2^n+1}+\frac{2-2^{n-1}}{2^n-1}\right)\mathrm{Tr}\left[H^{\otimes 2}\mathbb{F}\right]+\frac12\left(\frac{2+2^{n-1}}{2^n+1}-\frac{2-2^{n-1}}{2^n-1}\right)\mathrm{Tr}\left[H^{\otimes 2}\iden\right]\\
&=\left(\frac{2- 2^{-1}}{2^n-2^{-n}}\right)\mathrm{Tr}[H^2]+\left(\frac{2^{2n-1}-2}{2^{2n}-1}\right)\mathrm{Tr}[H]^2\\
&\geq \left(\frac{2- 2^{-1}}{2^n-2^{-n}}\right)\mathrm{Tr}[H^2].
\end{split}
\end{align}
Consequently, we obtain
\begin{equation}
\mathbb{E}\left(\min_{V}||C'(UV\otimes \mathbbm{1}_{n-1})C X-C'(U\otimes \mathbbm{1}_{n-1})C Y||_{\mathrm{F}}^2\right)\leq  \left(1-\frac{3}{2^{2n}-1}\right)\varepsilon^2\mathrm{Tr}(H^2)+O(\varepsilon^3).
\end{equation}
From
\begin{align}
\begin{split}
||X-Y||^2_{F}= \varepsilon^2\mathrm{Tr}(H^2)+O(\varepsilon^3),
\end{split}
\end{align}
we find
\begin{equation}
\mathbb{E}||X'-Y'||^2_{F}\leq  \left(1-\frac{3}{2^{2n}-1}\right)||X-Y||^2_{F}+O(\varepsilon^3).
\end{equation}
Notice, that for $X\neq Y$, we can always choose $\varepsilon$, such that $\mathrm{Tr}(H^2)=||H||_F^2=1$.
Therefore, we find that
\begin{align}
\begin{split}
\frac{W_{\mathrm{Fro},2}(\sigma*\delta_X,\sigma*\delta_Y)}{||X-Y||_F}&\leq \frac{(\mathbb{E}||X'-Y'||_F^2)^{\frac12}}{||X-Y||_F}\\
&\leq \frac{\sqrt{\left(1-\frac{3}{2^{2n}-1}\right)||X-Y||^2_{F}+O(\varepsilon^3})}{||X-Y||_F}\\
&\leq \sqrt{\left(1-\frac{3}{2^{2n}-1}\right)}+\sqrt{\frac{O(\varepsilon^3)}{\varepsilon^2+O(\varepsilon^3)}}\\
&\leq \sqrt{\left(1-\frac{3}{2^{2n}-1}\right)}+O(\sqrt{\varepsilon}),
\end{split}
\end{align}
with $\varepsilon\to 0$ as $X\to Y$.
Invoking Lemma~\ref{lemma:BHHsummaryofpathcoupling} $k$ times with $\eta:= \sqrt{(1-\frac{3}{2^{2n}-1})}$ implies 
\begin{align}
\begin{split}
W_{\mathrm{Fro},2}(\sigma^{*k},\mu_H)&=W_{\mathrm{Fro},2}(\sigma^{*k}*\delta_{\mathbbm{1}},\sigma*\mu_H)\\
&\leq \eta^{k}W_{\mathrm{Fro},2}(\delta_{\mathbbm{1}},\mu_H)\\
&\leq \eta^k (\mathbb{E}_U||\mathbbm{1}-U||_F^2)^{\frac12}\\
&=\sqrt{2}2^{n/2}\eta^k,
\end{split}
\end{align}
which is the statement of Lemma~\ref{lemma:wassersteinbound}.
\end{proof}
We can now prove Lemma~\ref{lemma:auxiliarywalk}:
\begin{proof}[Proof of Lemma~\ref{lemma:auxiliarywalk}]
	We use the following formula~\cite{brandao_local_2016}, which holds for any probability distribution $\nu$:
	\begin{equation}
	g(\nu,t)\leq 2 t W_{\mathrm{Fro},1}(\nu,\mu_{H})
	\end{equation}
	Combining this with Lemma~\ref{lemma:wassersteinbound} yields:
	\begin{align}
	\begin{split}\label{eq:boundonconvolutedsigma}
	g(\sigma^{*k},t)&\leq 2t W_{\mathrm{Fro},1}(\sigma^{*k},\mu_{H})\\
	&\leq 2 t W_{\mathrm{Fro},2}\big(\sigma^{*k},\mu_{H}\big)\\
	&\leq \left(1-\frac{3}{2^{2n}-1}\right)^{k/2}2\sqrt{2} t 2^{n/2},
	\end{split}
	\end{align}
	where the second inequality follows immediately from Jensen's inequality.
 $\sigma$ is a symmetric measure, that is, it is invariant under taking inverses.
	Consequently, the moment operators are hermitian, which implies
	\begin{equation}\label{eq:gapamplification}
g(\sigma^{*k},t)=	g(\sigma,t)^k.
	\end{equation}
	Notice that the inequality $g(\sigma^{*k},t)\leq g(\sigma,t)^k$ holds for all measure but not equality.
	We plug~\eqref{eq:gapamplification} into~\eqref{eq:boundonconvolutedsigma}, take the $k$-th square root and the limit $k\to \infty$.
	This yields
	\begin{equation}
	g(\sigma,t)\leq \left(1-\frac{3}{2^{2n}-1}\right)^{1/2}\leq 1-\frac12 \frac{3}{2^{2n}-1}.
	\end{equation}
\end{proof}

Before we can prove Theorem~\ref{thm:unconditionalgap}, we need another auxiliary bound.
More precisely, we show that a local random walk over Clifford generators quickly approximates the uniform measure on the Clifford group:
 Apply first a unitary $U_{1,2}\otimes U_{3,4}\otimes...$ and then a unitary $U_{2,3}\otimes U_{4,5}\otimes...$, where all $U_{i,i+1}$ are drawn uniformly from the Clifford group on $2$ qubits. 
We denote this distribution by $\nu^{\rm Cl, bw}_n$ 

The Clifford group is a finite group with polynomially bounded circuit complexity~\cite{aaronson2004improved,bravyi2021hadamard} and we can use this fact to apply the comparison technique~\cite{diaconis1993comparison}.
This gives the following bound:
\begin{lemma}\label{lemma:comparisonmethod}
	The following bound holds for all $t\geq 1$:
	\begin{equation}\label{eq:comparisonbound}
	\left|\left|M\left(\nu^{\rm Cl, bw}_n,t\right)-M(\mu_{\mathrm{Cl}},t)\right|\right|_{\infty}\leq 1-\frac{1}{2000 n^3}.
	\end{equation}
\end{lemma}
\begin{proof}
	Consider the averaging operator on the group algebra $T_{\nu}: L^2(\mathrm{Cl}_n)\to L^2(\mathrm{Cl}_n)$ defined by 
	\begin{equation}
	(T_{\nu}f)(g):=\int f(h^{-1}g)\mathrm{d} \nu(h).
	\end{equation}
	By~\cite{diaconis1993comparison}, we have
	\begin{equation}
	\lambda_2(T_{\sigma})\leq 1-\frac{\eta}{d^2},
	\end{equation}
	where $\eta$ is the probability of the least probable generator and $d$ is the minimal number of generators necessary to generate any group element.
	Here, $d=9n$ by~\cite{bravyi2021hadamard} and $\eta=1/|\mathrm{Cl}_2|n=1/24n$.
	By the Peter-Weyl theorem, the entire spectrum of the moment operators is contained in the spectrum of the averaging operator and thus we obtain Eq.~\eqref{eq:comparisonbound}.
\end{proof}
We remark that the comparison technique was previously applied to the Clifford group in Ref.~\cite{divinzeno_data_2001}.

We can now put everything together.
\begin{proof}[Proof of Theorem~\ref{thm:unconditionalgap}]
	For simplicity, we use the following notation:
	\begin{equation}
	P_{\mathrm{Cl},n}=M(\mu_{\mathrm{Cl}_n},t).
    \end{equation}
	We use 
	\begin{equation}
	\left|\left|M\left(\nu^{\rm Cl, bw}_n,t\right)^k-P_{\mathrm{Cl},n}\right|\right|_{\infty}\leq \left|\left|M\left(\nu^{\rm Cl, bw}_n,t\right)-P_{\mathrm{Cl},n}\right|\right|^k_{\infty}\leq \left(1-\frac{1}{2000 n^3}\right)^k.
	\end{equation}
	In particular,  we have the approximation
	\begin{equation}
		\left|\left|M\left(\nu^{\rm Cl, bw}_n,t\right)^{k_{n}}-P_{\mathrm{Cl},n}\right|\right|_{\infty}\leq \frac12 \frac{1}{2^{2n}-1}
	\end{equation}
	for a number of layers $k_n$ that we can choose as
	\begin{equation}\label{eq:k_n}
	k_n= 6000n^4.
	\end{equation}
	Therefore, by applying the triangle inequality twice, we obtain
	\begin{align}
	\begin{split}\label{eq:auxiliarygapbound}
	||M\left(\nu^{\rm Cl, bw}_n,t\right)^{k_{n}}M(\mu_{H,1}&,t)M\left(\nu^{\rm Cl, bw}_n,t\right)^{k_{n}}-P_{H}||_{\infty}\\
	&\leq ||P_{\mathrm{Cl},n}M(\mu_{H,1},t)P_{\mathrm{Cl},n}-P_{H}||_{\infty}+\frac{1}{2^{2n}-1}\\
&= ||M(\sigma,t)-P_{H}||_{\infty}+\frac{1}{2^{2n}-1}\\
&\leq 1-\frac12 \frac{3}{2^{2n}-1}+ \frac{1}{2^{2n}-1}\\
&\leq 1-\frac{1}{2(2^{2n}-1)},
\end{split}
	\end{align}
	where we used Lemma~\ref{lemma:auxiliarywalk} in the third inequality.

 $M\left(\nu^{\rm Cl, bw}_n,t\right)^{k_{n}}M(\mu_{H,1},t)M\left(\nu^{\rm Cl, bw}_n,t\right)^{k_{n}}$ is a product of $(2nk_n+1)$ orthogonal projectors.
 In the following we denote these projectors as $\iden-Q_{1},\ldots, \iden-Q_{2nk_n+1}$.
E.g. we set 
\begin{equation}
Q_1:=\iden-P_{\mathrm{Cl},2}\otimes \iden_{n-2}
\end{equation}
and $Q_{n k_n+1}:=\iden-P_{H,1}$.
The order of these \jh{labels} will not matter in the following.
As in Lemma~\ref{lemma:detectabilityandconverse}, we define a Hamiltonian $\tilde{H}:=\sum_{i=1}^{2 n k_n+1} Q_i$.
We can now relate the bound~\eqref{eq:auxiliarygapbound} to the gap of the Hamiltonian $\tilde{H}$ via the quantum union bound or converse detectability lemma (Lemma~\ref{lemma:detectabilityandconverse}):
 \begin{align}
 \begin{split}
||M\left(\nu^{\rm Cl, bw}_n,t\right)^{k_{n}}M(\mu_{H,1}&,t)M\left(\nu^{\rm Cl, bw}_n,t\right)^{k_{n}}-P_{H}||_{\infty}\\
&=\max_{|\psi^{\perp}\rangle}||M\left(\nu^{\rm Cl, bw}_n,t\right)^{k_{n}}M(\mu_{H,1},t)M\left(\nu^{\rm Cl, bw}_n,t\right)^{k_{n}}|\psi^{\perp}\rangle||_2\\
&=\max_{|\psi^{\perp}\rangle}\left|\left|\prod_{i=1}^{2 n k_n+1} (\iden -Q_i)|\psi^{\perp}\rangle\right|\right|_2\\
&\geq \sqrt{1-4\Delta(\tilde{H})}\\
&\geq 1-4\Delta(\tilde{H}).
\end{split}
 \end{align}
 Combining this with~\eqref{eq:auxiliarygapbound}, yields:
 \begin{equation}
  \frac{1}{8} \frac{1}{2^{2n}-1}\leq \Delta(\tilde{H}).
 \end{equation}
 We still need to relate the gap of $\tilde{H}$ to that of $H_{n,t}$.
 Using the operator inequalities $P_{\mathrm{Cl}}\geq P_H$ and $P_{H,1}\geq P_{H,2}$, we obtain
 \begin{align}
 \begin{split}
 \tilde{H}&= 2k_n\sum_{i=1}^n(\iden- \iden_{i-1}\otimes P_{\mathrm{Cl},2}\otimes \iden_{n-i-1}) + (\iden-P_{H,1}\otimes \iden_{n-1})\\
 &\leq  2k_n\sum_{i=1}^n(\iden- \iden_{i-1}\otimes P_{H,2}\otimes \iden_{n-i-1}) + (\iden-P_{H,1}\otimes \iden_{n-1})\\
 &\leq (2k_n+1)\sum_{i=1}^n(\iden- \iden_{i-1}\otimes P_{H,2}\otimes \iden_{n-i-1})\\
 &=(2k_n+1)H_{n,t}.
 \end{split}
 \end{align}
 \jh{Moreover}, notice that $\tilde{H}$ and $H_{n,t}$ have the same ground state space with ground state energy $0$ as the gate set consisting of all $2$-local Clifford unitaries plus single qubit unitaries are universal~\cite{nebe_clifford_2001}.
 Indeed, it was proven in Ref.~\cite{harrow2009random} that the eigenvalue $1$ eigenspace of a moment operator for a probability distribution with universal support equals the image of the Haar random moment operator.
 Consequently, we find
 \begin{equation}
\Delta(H_{n,t})\geq \frac{1}{8(2k_n+1)(2^{2n}-1)}\geq 120000^{-1}n^{-4}2^{-2n}.
 \end{equation}
 Together with Eq.~\eqref{eq:gapreduction}, this implies Theorem~\ref{thm:unconditionalgap}.
\end{proof}

\section{Outlook: Gaps from representations of the Clifford group?}
In this work, we improved the scaling of the design depth for local random quantum circuits over qubits.
The lion's share of the improvement comes from the near-optimal convergence of an auxiliary random walk that interleaves global random Clifford unitaries and single qubit Haar random unitaries. 

The key open question, posed in Ref.~\cite{brandao_local_2016}, is to either prove or disprove a subexponential unconditional gap.
An unconditional gap that scales inverse polynomially in $n$ would imply a robust version of the Brown-Susskind conjecture for random quantum circuits. 
For approaches based on techniques from harmonic analysis, see Ref.~\cite{haferkamp2021improved}.
The techniques presented in this paper open up another potential path towards this goal: It suffices to prove such a gap for the auxiliary walk defined in Eq.~\eqref{eq:auxiliarywalk}.

The commutant of tensor powers of the Clifford group was studied in Ref.~\cite{gross2021schur,montealegre_howe_2019,montealegre2022duality} and it was proven in Ref.~\cite{haferkamp2020quantum} that the action of a single qubit Haar average has a strong shrinking effect on a natural basis of the Clifford commutant labeled by Lagrangian subspaces of $\mathbb{Z}_2^{2t}$.
Therefore, Lemma 13 in Ref.~\cite{haferkamp2020quantum} can be viewed as a sanity check for the constant spectral gap conjecture.

To prove the generation of state designs in linear depth, we do not require a spectral gap of the entire moment operator.
Indeed, it would suffice to show a spectral gap of the moment operator restricted to the endomorphisms of the symmetric subspace $\mathrm{End}(S^t[(\mathbb{C}^2)^{\otimes n}])$:
\begin{equation}
\left|\left|(M(\nu_n,t)-P_H)\large|_{\mathrm{vec}\left(\mathrm{End}(S^t[(\mathbb{C}^2)^{\otimes n}])\right)}\right|\right|_{\infty}\leq 1-\mathrm{poly}^{-1}(n)?
\end{equation}
To characterize the commutant of powers of the Clifford group on the endomorphisms of the symmetric subspace, one would need to understand how the symmetric subspace decomposes into irreducible representations of the Clifford group.

\section{Acknowledgements}
I am grateful to Nick Hunter-Jones, Richard Kueng and Felipe Montealegre-Mora for fruitful discussions.
Moreover, I want to thank Jens Eisert for detailed comments on the manuscript.
This work was funded by the foundational questions institute (FQXi).

\bibliographystyle{plainurl}

\appendix
\section{Proof of Observation~\ref{observation:improvedreduction}}\label{app:improvedreduction}
In this appendix we show how Observation~\ref{observation:improvedreduction} follows from bounds on Lambert's W function.
The improvement obtained in Observation~\ref{observation:improvedreduction} makes the difference between a design depth of $O(nt^{6})$ and a design depth of $O(nt^{5+o(1)})$.
\begin{proof}
	Eq.~\eqref{eq:BHHreductioninequality} is implied by 
	\begin{equation}\label{eq:logofreduction}
	l -\frac{1}{2\ln(2)}\ln(l)\geq 2\log_2(4t).
	\end{equation}
	First, we consider the equation $y-a\ln(y)=b$.
	This is equivalent to
	\begin{equation}\label{eq:manipulatedBHHcondition}
	e^{-\frac{y}{a}}\left(-\frac{y}{a}\right)=-\frac{e^{-\frac{b}{a}}}{a}.
	\end{equation}
If the right hand side of Eq.~\eqref{eq:manipulatedBHHcondition} is larger than $-1/e$, this can be solved by
	\begin{equation}
	y=-aW_{-1}\left(-\frac{e^{-b/a}}{a}\right),
	\end{equation}
	where $W_{-1}$ denotes one of the branches~\cite{chatzigeorgiou2013bounds} of Lambert's W function (also called log product function).
	By definition, all branches of Lambert's W function solve the equation
	\begin{equation}
	y e^y=x
	\end{equation}
	for $y$.
	
	We can now use the following bound for $x>0$ proven in Ref.~\cite{chatzigeorgiou2013bounds}:
	\begin{equation}
	W_{-1}(e^{-x-1})\geq -1-(2x)^{\frac12}-x.
	\end{equation}
	With $x=b/a+\ln(a)-1$, $b=2\log_2(4t)$ and $a=(2\ln(2))^{-1}$, this implies
	\begin{align}
	\begin{split}
	y&\leq -a \left(-1-\sqrt{2\left(\frac{b}{a}+\ln(a)-1\right)}-\frac{b}{a}-\ln(a)+1\right)\\
	&\leq 2\log_2(4t)+1.5\sqrt{\log_2(4t)}.
	\end{split}
	\end{align}
The LHS of Eq.~\eqref{eq:logofreduction} is monotone in $l$.
Hence, choosing
	\begin{equation}
	l=\left\lceil 2\log_2(4t)+1.5\sqrt{\log_2(4t)}\right\rceil
	\end{equation}
	implies Eq.~\eqref{eq:logofreduction} and thus, by Nachtergaele's bound, Eq.~\eqref{eq:improvedreduction}.
\end{proof}
%

\end{document}